\documentclass[12pt]{article}
\usepackage{amsfonts}
\usepackage{amssymb, amsthm, amsbsy, amscd, amsfonts, latexsym, euscript,epsf,stackrel}
\usepackage{amscd}
\usepackage{amsmath}
\usepackage{authblk}

\usepackage[utf8]{inputenc}
\usepackage[T2A]{fontenc}

\usepackage[usenames]{color}

\newcommand{\BEQ}{\begin{equation}}
\newcommand{\EEQ}{\end{equation}}
\def\bea{\begin{eqnarray}}
\def\eea{\end{eqnarray}}
\def\nn{\nonumber}

\theoremstyle{plain}
\newtheorem{Th}{Theorem}[section]
\newtheorem{Lem}[Th]{Lemma}

\newtheorem{Hyp}[Th]{Conjecture}

\newtheorem{Rem}[Th]{Remark}

\theoremstyle{definition}

\newtheorem*{ack}{Acknowledgements}

\def\bea{\begin{eqnarray}}
\def\eea{\end{eqnarray}}

\def\bes{\begin{equation*} \begin{split}}
\def\ees{\end{split} \end{equation*}}

\def\CC{{\mathbb{ C}}}

\newcommand{\sign}{\text{sign}}

\usepackage{graphicx}

\graphicspath{{pics/}}

\usepackage{hyperref}
\hypersetup{colorlinks=true,linkcolor=blue}

\begin{document}

\title{On a family of Poisson brackets on $\mathfrak{gl}_n$ compatible with the Sklyanin bracket
}

\author[1]{Vladimir V. Sokolov\footnote{vsokolov1952@gmail.com}}
\author[2,3]{Dmitry V. Talalaev\footnote{dtalalaev@yandex.ru}}
\affil[1]{\small Higher School of Modern Mathematics MIPT, 141701,  Moscow,   Russia , 
}

\affil[2]{\small Lomonosov Moscow State University, 
119991, Moscow, Russia.
}
\affil[3]{\small Center of Integrable Systems, Demidov Yaroslavl State University, Yaroslavl, Russia, 150003, Sovetskaya Str. 14
}

\date{}

\maketitle

\abstract{In this paper, we study a family of compatible quadratic Poisson brackets on  $\mathfrak{gl}_n$, generalizing the Sklyanin one. For any of the brackets in the family, the argument shift determines the compatible linear bracket. The main interest for us is the use of the bi-Hamiltonian formalism for some pencils from this family, as a method for constructing involutive subalgebras for a linear bracket starting by the center of the quadratic bracket. We give some interesting examples of families of this type.

We construct the centers of the corresponding quadratic brackets using the antidiagonal principal minors of the Lax matrix. Special attention should be paid to the condition of the log-canonicity of the brackets of these minors with all the generators of the Poisson algebra of the family under consideration. A similar property arises in the context of Poisson structures consistent with cluster transformations.
}

\vskip 5mm
~\\
\tableofcontents
\vskip 2cm
\section{Introduction}
Let us consider the vector space $V=\mathbb{C}^n$. The space of $n\times n$ matrices we will associate with the space $End(V)$. A basis in this space is given by matrix unities $E_{ij}$. Here and below
$e_{ij}$ - the full set of generators of the Lie algebra $\mathfrak{gl}_n,$ 
the Lax oeprator has the form 
\bea
\label{Lax}
L=\sum_{ij} E_{ij}\otimes e_{ij},
\eea
its tensor copies are expressed as:
\bea
L_1=L\otimes 1;\qquad L_2=1\otimes L.\nn
\eea
In this paper, we study Poisson structures of the form:
\bea
\label{PoissQuad2}
\{L\otimes L\}=A L_1 L_2+L_1 B L_2+L_2 C L_1 +L_1 L_2 D,
\eea
where $A,B,C,D\in End(V)^{\otimes^2}.$ A special case of such a bracket is the Sklyanin bracket
\bea
\label{SklBr}
\{L\otimes L\}=[R, L\otimes L],
\eea
which appeared in the framework of the quantum inverse scattering method (see, for example, in \cite{Skl}). It is essentially related to the theory of Lie-Poisson groups. The general class of Poisson brackets of the form \eqref{PoissQuad2} appeared in the works of Semenov-Tian-Shansky \cite{StSS}, \cite{StS} and Suris \cite{Suris1}, \cite{Suris2}. Such brackets were obtained by reductions from the standard $r$-matrix brackets on the dubble $\mathfrak{gl}_n\oplus\mathfrak{gl}_n$. It is not difficult to verify that the right hand side of  bracket \eqref{PoissQuad2} can be reduced by the skew-symmetrization procedure to the form (\ref{norm})
\bea
\label{norm}
C=-B^{*},\qquad  A^{*}=-A, \qquad D^{*}=-D.
\eea
The book \cite[section 2.5]{Suris2} provides sufficient conditions for fulfilling the Jacobi identity for such brackets.

The origin of such Poisson structures is largely due to the theory of quantum groups \cite{FRT}, that is, algebras of the RTT type and their close relatives RE-algebras. Recall that the Drinfeld $R$- matrix is the expression
\bea
R=\sum_{i<j}(q-q^{-1})E_{ii}\otimes E_{jj}+\sum_{i,j}q^{\delta(i,j)}E_{ij}\otimes E_{ji} \in Mat_{n}\otimes Mat_{n}.\nn
\eea
It determines the solution of the quantum Yang-Baxter equation
\bea
R_{12} R_{13} R_{23}=R_{23} R_{13} R_{12}.\nn
\eea
An RTT algebra is defined as a quotient algebra of a free algebra  $T^{\cdot}(t_{ij})$  where $i,j\in(1,\ldots,n),$ by the quadratic relations
\bea
R T_1 T_2=T_1 T_2 R.\nn
\eea
Here $T$ is a generating matrix for RTT-algebra elements
 \bea
 T=\sum E_{ij}\otimes t_{ij}; \qquad  T_1=\sum E_{ij}\otimes 1 \otimes t_{ij};\qquad  L_2=\sum 1\otimes E_{ij} \otimes t_{ij}.\nn
 \eea
 RE-algebras were first introduced by Cherednik\cite{Cher} when studying scattering processes with a boundary.
Such an algebra $\mathcal{A}$ is defined using generators $a_{ij}$ and quadratic relations of the form:
\bea
\label{RE}
RL_1RL_1=L_1RL_1R \in Mat_{n}^{\otimes 2}\otimes \mathcal{A},\nn
\eea
where $L$ is a generating matrix of $a_{ij}$:
 \bea
 L=\sum E_{ij}\otimes a_{ij},\qquad
 L_1=\sum E_{ij}\otimes 1 \otimes a_{ij}.\nn
 \eea
Drinfeld $R$-matrix can be expanded in the parameter $\hbar=2 \log(q)$ at zero:
\bea
R=P+\hbar \rho +O(\hbar^2),\nn
\eea
where
\bea
\label{rho}
\rho=\frac 1 2\sum_i E_{ii}\otimes E_{ii}+\sum_{i<j} E_{ii}\otimes E_{jj},\qquad P=\sum_{ij}E_{ij}\otimes E_{ji}.
\eea
Let us denote
\bea
\label{r}
r=\rho P=\frac 1 2\sum_{i} E_{ii}\otimes E_{ii}+\sum_{i<j}E_{ij}\otimes E_{ji},
\eea
and introduce an involution operation on the elements $a\in Mat_n\otimes Mat_n$ by
\bea
\qquad a^*=PaP,\nn
\eea
which acts as a permutation on matrix generators
\bea
\left(E_{ij}\otimes E_{kl}\right)^{*}=E_{kl}\otimes E_{ij}.\nn
\eea
Then
\bea
r^*=P r P=P\rho;\qquad  (\rho+\rho^*)=E,\qquad(r+r^*)=P.\nn
\eea
The RTT-algebra in the limit $\hbar\rightarrow 0$  produces the Sklyanin bracket
\bea
\label{Skl}
\{L\otimes L\}=[r,L\otimes L].
\eea
A similar limit for the RE-algebra gives a generalized quadratic bracket
\bea 
\label{REclass}
\{L\otimes L\}=r L_1 L_2+ L_2 r^* L_1-L_1 r L_2 - L_1 L_2 r^*.
\eea

Together with the quadratic brackets \eqref{PoissQuad2}, we consider the linear Poisson brackets, which are obtained from the quadratic ones by  the argument shift. This procedure consists in replacing the coordinates $L \mapsto L+\lambda E$ in the phase space. In the new coordinates, the bracket \eqref{PoissQuad2} has the form
\bea
\label{pencil}
\{L\otimes L\}_{\lambda}=\{L\otimes L\}+\lambda \{L\otimes L\}_{lin}+\lambda^2 \{L\otimes L\}_{const}.
\eea
It is easy to see that
\bea
\label{PoissLin2}
\{L\otimes L\}_{lin}=[A+C,L_1]+[A+B,L_2].
\eea
The coefficient of $\lambda^2$ vanishes if and only if
$$
A+B+C+D=0.
$$
Under this condition, the bracket \eqref{PoissLin2} satisfies the Jacobi identity and is a linear Poisson bracket compatible with \eqref{PoissQuad2} \cite{magri1} .

For example, the linear bracket generated by the argument shift of the bracket \eqref{REclass} coincides with the standard $\mathfrak{gl}_n$ - bracket.

\section{The general family of quadratic brackets}

Consider the following specialization of quadratic brackets (\ref{PoissQuad2}):
\bea
\label{quad-bracket}
\{L\otimes L\}=(A+S)L_1 L_2 +L_1(B-S)L_2+L_2(-B-S)L_1+L_1 L_2(-A+S),
\eea
where coefficients have the form
\bea\label{coe}
A&=&- w \sum_{i>j} (E_{ij}\otimes E_{ji} -E_{ji}\otimes E_{ij})+\sum_{i>j} a_{ij}(E_{ii}\otimes E_{jj}-E_{jj}\otimes E_{ii});
\nn\\ 
B&=&\sum_{i> j} b_{ij}(E_{ii}\otimes E_{jj}+E_{jj}\otimes E_{ii})+\sum_i b_{ii} E_{ii}\otimes E_{ii};\\
S&=&\sum_{i>j} s_{ij}(E_{ii}\otimes E_{jj}-E_{jj}\otimes E_{ii}) \nn.
\eea
Here $w, a_{ij},  b_{ij} $ и $s_{ij}$ are some constants.
The components of such a bracket satisfy the relations
\bea
A^*=-A,\quad ~B^*=B, \quad S^*=-S.\nn
\eea
If $w=1/2$ and the remaining constants are equal to zero, formulas \eqref{quad-bracket}, \eqref{coe} give the Sklyanin bracket written in the normalized form \eqref{norm}.

\begin{Th}
The brackets \eqref{quad-bracket}, \eqref{coe} for any parameter values are Poisson brackets.
\end{Th}
\begin{proof}
The brackets of the family under consideration are local, that is, $\{e_{ij}, e_{kl}\}$ contains only products of the elements $\{e_{pq}, e_{rs}\}$, where all indices belong to the set $\{i,j,k,l\}$. Therefore, when checking the Jacobi identity, it suffices to consider $n=6$. In this case, the fact that the Jacobi identity is fulfilled can be verified by direct calculation.
\end{proof}

The family of corresponding linear 
$r$-matrix brackets (\ref{PoissLin2}) in the case of \eqref{quad-bracket}, \eqref{coe} specializes as follows:
\bea
\label{PoissLin1}
\{L\otimes L\}_{lin}=[A-B,L_1]+[A+B,L_2].
\eea
Notice that the tensor $S$ is missing in the coefficients of the linear Poisson bracket.  
 
Let us introduce the following notation:
\bea
c_{i,j}=a_{ij}+b_{ij},\qquad
h_{k,i}=(c_{k,i}-c_{k,i+1}).\nn
\eea
It turns out that the linear bracket (\ref{PoissLin1}) has a natural interpretation in terms of extensions of Lie algebras:
\begin{Th} \label{ThLinPoiss} 
The  bracket \eqref{PoissLin1} has the following properties:
\begin{itemize}
\item being restricted to   $\mathfrak{n}+$ и $\mathfrak{n}_-$ it coincides with the standard Poisson bracket 
$\{X, Y\}_{\mathfrak{gl}}$, multiplied by 
$2 w$ and $-2 w$,  respectively;

 \item $\{\mathfrak{n}+,\mathfrak{n}-\}_{lin} = 0$;

\item the action of the Cartan subalgebra $\mathfrak{h}$ on $\mathfrak{n_+}\oplus \mathfrak{n_-}$ is given by the formulas
\bea
\{u_{k,k}, u_{i,i+1}\}_{lin}&=&( w  (\delta_{k,i}-\delta_{k,i+1})+h_{k,i} ) u_{i,i+1},\nn\\
\{u_{k,k}, u_{i+1,i}\}_{lin}&=&( w (\delta_{k,i}-\delta_{k,i+1})-h_{k,i})  u_{i+1,i}, 
\eea
where \quad $k=1,\dots, n, \,\,  i=1,\dots, n-1.$ 
\end{itemize}
\end{Th}

In further considerations of our work, the antidiagonal minors of the $L$ matrix play a key role. Such expressions arise in many integrable Hamiltonian models  related to the algebra $\mathfrak{gl}_n$, including varieties of the full Toda system \cite{DLNT}. The same functions play an important role in describing the cluster structure on flag varieties \cite{GSV}. In the case of quantum integrable systems, quasi-determinants for RE-algebras similar to these minors were considered in \cite{T}, a similar approach for $U_q(\mathfrak{sl}_n)$ was used in \cite{Skr}. 
In the case of the Sklyanin bracket, antidiagonal minors appeared in \cite{MarFok}.

Denote by $\Delta^{+}_k$ (respectively $\Delta^{-}_k$) the upper-right (respectively lower-left) antidiagonal minor of size $k$ of the matrix $L$. The following property is characteristic for quadratic brackets
\eqref{quad-bracket}, \eqref{coe}.

\begin{Th} For any bracket \eqref{quad-bracket}, \eqref{coe} the following relations hold:
\bea
\{e_{ij}, \Delta^{+}_k\}&=& \alpha_{ij}(k, n)\, e_{ij}\, \Delta^{+}_k,\nn\\
\{e_{ij}, \Delta^{-}_k\}&=& \beta_{ij}(k, n)\, e_{ij}\, \Delta^{-}_k,
\label{log-can}
\eea
where $\alpha_{ij}(k, n)$ and $\beta_{ij}(k, n)$ are some constants that depend linearly on the coefficients of the bracket.
\end{Th}

\begin{proof}
A general bracket of the family is a linear combination of the Sklyanin bracket and a log-canonical bracket. 
First of all, we show that the statement of the theorem holds for the Sklyanin bracket. Namely,
the coefficients $\alpha_{ij}$ can be represented by the matrix
 \begin{figure}[h!]
\center
\includegraphics[width=53mm]{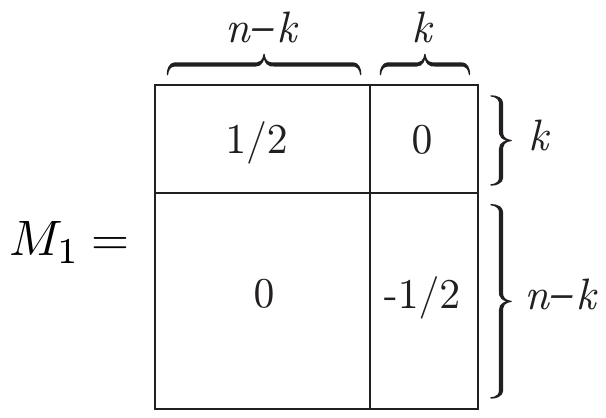}
\end{figure}

\noindent
whose entries are divided into blocks by rows: from $1$-st to $k$-th and from $k+1$-st to $n$-th, and by columns: from $1$-st to $n-k$-th, and from $n-k+1$-st to $n$-th.
The values in upper triangular blocks are a consequence of the lemma 4.7 \cite{GSV}.
The zero in the lower-left corner of the matrix follows from the fact that all elements of the lower-left corner commute with all elements of the upper-right corner by virtue of the Sklyanin commutation relations. 

To prove the relations \eqref{log-can} for the lower-left minors, we use the symmetry of the Sklyanin bracket with respect to the transposition involution
\bea
\{f^t,g^t\}=\{f,g\}^t.\nn
\eea
 This implies the equality 
 \bea
 \beta_{ij}(k,n)=\alpha_{ji}(k,n).\nn
 \eea
It is not difficult to verify that the answer for the Sklyanin bracket, defined by the matrix $M_1,$  can be written as
\bea
\alpha_{ij}(k,n)&=&1/4 \Big(\sign(k - i) + \sign(n - k - j) + 
 \delta_{k, i} + \delta_{n - k, j}\Big), \nn\\
 \beta_{ij}(k,n)&=&1/4 \Big(\sign(k - j) + \sign(n - k - i) + 
 \delta_{k, j} + \delta_{n - k, i}\Big). 
\label{albe}
\eea

Adding a log-canonical bracket to the Sklyanin bracket changes $\alpha_{ij}(k,n)$ and $\beta_{ij}(k,n)$ but does not change the general form of the relations (\ref{log-can}).
\end{proof}
\begin{Rem}
We believe that any bracket \eqref{PoissQuad2} satisfying the conditions \eqref{log-can} can be written in the form \eqref{quad-bracket}, \eqref{coe}. This thesis has been verified for $n\le 6$.
\end{Rem}

The formulas \eqref{log-can} explain the existence of central elements of the form
\begin{equation} \label{casim}
\prod_{k=1}^n \Big(\Delta^{+}_k\Big)^{p_k} \Big(\Delta^{-}_k\Big)^{q_k}, \qquad p_k, q_k \in \CC
\end{equation}
for many brackets \eqref{quad-bracket}, \eqref{coe}.

\begin{Rem}\label{Rem1.6}
The polynomials
$$
\Delta^{+}_1, \dots , \Delta^{+}_{n-1}, \, \Delta^{-}_1, \dots \dots , \Delta^{-}_{n}
$$
are functionally independent. Indeed, each subsequent polynomial from this set contains elements of $L$ that are not present in all the previous ones.
\end{Rem}

\section{4-dimensional family}

In this section we will limit ourselves to considering the subfamily of quadratic brackets of the form
\begin{equation} \label{4brackets}
\{e_{ij}, e_{kl}\}=\frac{z_1}{2} \Big(\sign(k - i) + \sign(l - j)\Big) e_{i l} e_{k j}+ 
\frac{z_2}{2} \Big(\delta_{i l} - \delta_{j k}\Big) e_{i j} e_{k l}+
\end{equation} 
$$ z_3 \Big(\sign(i - k) - \sign(j - k) - \sign(i - l) + \sign(j - l)\Big) e_{i j} e_{k l}+z_4 \Big(\sign(i - k) - 
\sign(j - l) \Big) e_{i j} e_{k l}.
$$ 
All such brackets are representable in the form 
(\ref{PoissQuad2})
\bea
\{L\otimes L\}_2=A L_1 L_2+L_1 B L_2+L_2 C L_1 +L_1 L_2 D.\nn
\eea
In the case $z_1=1, z_2=z_3=z_4=0$ the formula \eqref{4brackets} defines the Sklyanin bracket \eqref{Skl}
\bea
\label{r1_2}
A=-D = r=\frac 1 2\sum_{i} E_{ii}\otimes E_{ii}+\sum_{i<j}E_{ij}\otimes E_{ji}; \qquad B=C=0;
\eea
for the bracket with $z_2=1, z_1=z_3=z_4=0$ we get
\bea
\label{r2_2}
B=-C =- \frac 1 2 \sum_i E_{ii}\otimes E_{ii}, \qquad A=D=0;
\eea
the values  $z_3=1, z_1=z_2=z_4=0$ correspond to
\bea
\label{r3_2}
A=D=-B=-C=-\rho=- \frac 1 2 \sum_i E_{ii}\otimes E_{ii}-\sum_{i<j} E_{ii}\otimes E_{jj};
\eea
in the case $z_4=1, z_1=z_2=z_3=0$ we obtain
\bea
\label{r4_2}
A=-D=\sum_{i>j} e_{ii}\otimes e_{jj} - e_{jj}\otimes e_{ii}, \qquad  B=-C=\sum_{i,j} e_{ii}\otimes e_{jj}.
\eea

In the next section, using explicit coordinate formulas \eqref{4brackets} for the brackets of this four-dimensional family, we find their center in some cases.   Apparently, for other brackets from the family \eqref{quad-bracket} there is no such a simple 
 coordinate representation.

The linear Poisson brackets \eqref{PoissLin2}, which are obtained from the brackets \eqref{quad-bracket} by the shift of argument, are described in
Theorem \ref{ThLinPoiss}. Namely, for any $z_1,z_2, z_4$  we have $\{\mathfrak{n_+}, \mathfrak{n_-}\}_{lin}=0.$ For the elements of $\mathfrak{n_+}$ we get $\{X, Y\}=z_1 \{X, Y\}_{\mathfrak{gl}}$, and for the elements of $\mathfrak{n_-}$ the formula $\{X, Y\}=-z_1 \{X, Y\}_{\mathfrak{gl}}$ holds.

The dependence of the bracket on the parameters $z_2,z_4$ is contained in the formulas for the action of a commutative subalgebra
 $e_{ii}, \, i=1,\dots, n$ on $\mathfrak{n_+}$ and $\mathfrak{n_-}$. This action is uniquely determined by the formulas
$$
\{e_{kk}, e_{i,i+1}\}=\Big(\frac{1}{2} (z_1 - z_2 + 2 z_4)\, \delta_{k, i} +  \frac{1}{2} (-z_1 + z_2 + 2 z_4)\, \delta_{k, i + 1}\Big) e_{i,i+1}
$$
and
$$
\{e_{kk}, e_{i+1,i}\}=\Big(\frac{1}{2} (-z_1 - z_2 - 2 z_4)\, \delta_{k, i} +  \frac{1}{2} (z_1 + z_2 - 2 z_4)\, \delta_{k, i + 1}\Big) e_{i+1,i}.
$$
For some parameter sets, such linear brackets \eqref{PoissLin2} are well known. For example,
in the case of $z_1=z_2=1, \, z_4=0$
the bracket is related to the decomposition
\bea \label{bn}
\mathfrak{gl}_n = \mathfrak{n_+}\oplus \mathfrak{b}_-
\eea
and is defined by the formula
$$
\{X, Y\}=\frac{1}{2}\Big(\{R(X), Y\}_{gl}+\{X, R(Y)\}_{gl} \Big),
$$
where $R$  is the difference of the projectors corresponding to the decomposition of \eqref{bn}.

Let us note that the linear brackets of the considered family do not depend on $z_3$. 
The presence of the parameter $z_3$ in the quadratic brackets \eqref{4brackets} turns out to be important when constructing Poisson reductions in the case of different versions of the Toda chain. Some examples are given at the end of the next section.


\section{Centers of the quadratic  brackets}

For Poisson structures from the family \eqref{4brackets}  with integers $z_1,z_2,z_3$ and $z_4=0$, the centers of rings of rational functions on $\mathfrak{gl}_n^*$ can be explicitly found using the ansatz \eqref{casim}. Here are some examples.

A direct consequence of the relations (\ref{albe}) is the following
\begin{Th}
\label{centr-Skl}
The functions $\frac{\Delta^{+}_k}{\Delta^{-}_{n-k}}$ for all $k$  and ${\rm det}\,L$
belong to the center of the Sklyanin bracket.
\end{Th}

\begin{Th}
\label{centr-1100}
In the case of the bracket $z_1=z_2=1,\, z_3=z_4=0$ the functions $\frac{\Delta^{-}_k}{\Delta^{-}_{n-k}}$, $ \Delta^{+}_k \Delta^{+}_{n-k}$ and  ${\rm det}\,L$ belong to the center of the Poisson algebra. 
\end{Th}

\begin{proof}
Let us show that the coefficients (\ref{log-can}) in the case $z_1=z_2=1,\, z_3=z_4=0$ have the form
\bea
 \alpha_{ij}(k,n)&=& 1/4 \Big(\sign(k - i) -\sign(k - j) - \sign(n - k - i) + 
   \sign(n - k - j)\nn\\
& +& \delta_{k, i} - \delta_{k, j} - \delta_{
    n - k, i} + \delta_{n - k, j}\Big),\nn\\
\beta_{ij}(k,n) &=& 1/4 \Big(\sign(k - i) +\sign(k - j) + \sign(n - k - i) + 
   \sign(n - k - j) \nn\\
&+& \delta_{k, i} + \delta_{k, j} + \delta_{
    n - k, i} + \delta_{n - k, j}\Big).\nn
\eea
The bracket under consideration is the sum of the Sklyanin bracket and the bracket of the form
\bea
\{e_{ij},e_{kl}\}_2=\frac 1 2 \Big(\delta_{i l} - \delta_{j k}\Big) e_{i j} e_{k l}.\nn
\eea
For the latter bracket we have
\bea
\{e_{ij}, \Delta^{+}_k\}_2= \alpha^{(2)}_{ij}e_{ij} \Delta^{+}_k,
\eea
where the coefficients $\alpha^{(2)}_{ij}$ form the following matrix $M_2$:
 \begin{figure}[h!]
\center
\includegraphics[width=53mm]{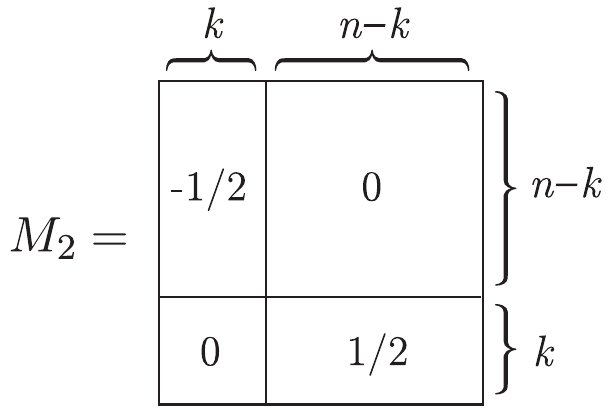}
\end{figure}

\noindent Indeed, consider the first term 
\bea
\{e_{ij},e_{kl}\}_{2,1}=\frac 1 2\delta_{il}e_{ij} e_{kl}\nn
\eea
of the bracket.
It is nontrivial only if the row number of the generator $e_{ij}$ coincides with the column number of the generator $e_{kl}$. Let us now calculate the bracket $e_{ij}$ with the upper-right minor $\Delta_{k}^+$. We introduce notation for the index sets $I=(1,\ldots,k)$ and $J=(n-k+1,n).$ Note that if $i\notin J,$ then the bracket is equal to zero. If, on the contrary, $i\in J$, then we use the decomposition of $\Delta_{k}^+$ by the $i$-th column
\bea
\Delta_{k}^+=\sum_{s=1}^k(-1)^{n-i+k+s}e_{si}\Delta_{J \setminus i}^{I \setminus s}\nn
\eea
to get
\bea
\{e_{ij},\Delta_k^+\}_{2,1}=\{e_{ij},\Delta^I_J\}_{2,1}=\sum_{s=1}^k(-1)^{n-i+k+s}\{e_{ij},e_{si}\}_{2,1}\Delta_{J \setminus i}^{I \setminus s}=\frac 1 2 e_{ij} \Delta_k^+.\nn
\eea
A similar calculation for the second term of the bracket yields
\bea
\{e_{ij},\Delta_k^+\}_{2,2}=-\frac 1 2 e_{ij} \Delta_k^+\nn
\eea
if and only if $j\in I.$

The coefficients for the sum of brackets $\{\cdot,\cdot\}_{(1,1)}=\{\cdot,\cdot\}+\{\cdot,\cdot\}_{2}$  can be represented by the matrix $M_1+M_2$. If  $k\le n/2$ this matrix is shown in Fig. \ref{pic-m1+m2}
\begin{figure}[h!]
\center
\includegraphics[width=4cm]{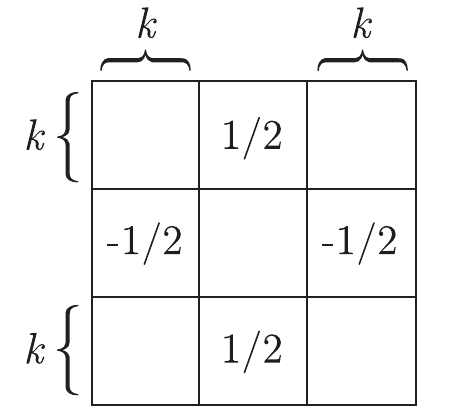}
\caption{}
\label{pic-m1+m2}
\end{figure}
 and for $k> n/2$ - in Fig. \ref{pic-m1+m2b}.  
 \begin{figure}[h!]
\center
\includegraphics[width=4cm]{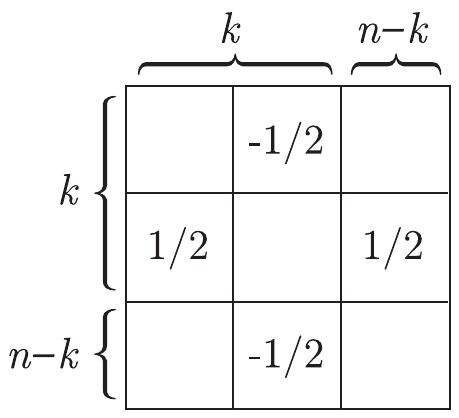}
\caption{}
\label{pic-m1+m2b}
\end{figure}

Note that the transposition operation $e_{ij}^t$ is an anti-involution of the bracket $\{\cdot,\cdot\}_{2}$ , so when calculating
$$\{e_{ij}, \Delta^{-}_k\}_2= \beta^{(2)}_{ij} e_{ij} \Delta^{-}_k,$$
we obtain
\bea
\beta^{(2)}=-M_2^t.\nn
\eea
Thus, the coefficients $\beta_{ij}$ for the brackets $\{\cdot,\cdot\}_{(1,1)}$ are defined by the matrix
\bea
M_1^t-M_2^t.\nn
\eea
This matrix for $k\le n/2$ has the form Fig. \ref{pic-m1-m2}
\begin{figure}[h!]
\center
\includegraphics[width=4cm]{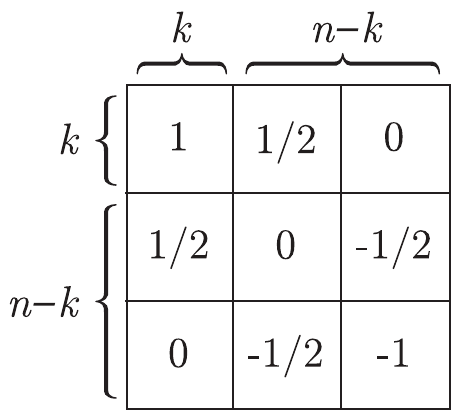}
\caption{}
\label{pic-m1-m2}
\end{figure}
and for  $k> n/2$  - on Fig. \ref{pic-m1-m2b}.
\begin{figure}[h!]
\center
\includegraphics[width=4cm]{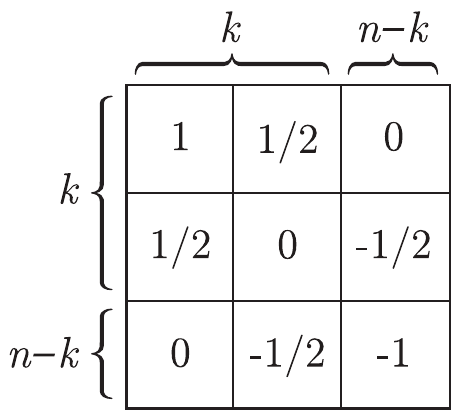}
\caption{}
\label{pic-m1-m2b}
\end{figure}
It follows directly from this that the functions $\frac{\Delta^{-}_k}{\Delta^{-}_{n-k}}$, $ \Delta^{+}_k \Delta^{+}_{n-k}$ together with the function ${\rm det}\,L$ lie in the center of the bracket under consideration.
\end{proof}

Similar calculations can be used to obtain centers in some other cases.
\begin{Rem}\label{centr-1110}
If $z_1=z_2=z_3=1,\, z_4=0$ then $\frac{\Delta^{+}_k}{\Delta^{+}_{n-k+1}}$, $\Delta^{+}_{1}$, 
 $\Delta^{+}_{n}$ and 
$\frac{\Delta^{-}_k}{\Delta^{-}_{n-k}}$  belong to the center. In this case there is a Poisson reduction procedure leading to a bi-Hamiltonian structure for the Constant-Toda system.  \cite{Kost}.
\end{Rem}
\begin{Rem}
\label{centr-11-10}
In the case $z_1=z_2=1, z_3=-1, z_4=0$ the functions $\frac{\Delta^{+}_k}{\Delta^{+}_{n-k-1}}$, $\Delta^{+}_{n-1}$, 
 $\Delta^{+}_{n}$ and $\frac{\Delta^{-}_k}{\Delta^{-}_{n-k}}$ lie in the center.
\end{Rem}
\begin{Rem}
In the case $k_1=2, k_3=-1, k_2=k_4=0$ the bracket admits a Poisson reduction leading to a bi-Hamiltonian structure for the standard tridiagonal open Toda chain. In this case, the center is generated by functions of the form \eqref{casim} with a rather complicated set of integer exponents $p_i, q_i.$
\end{Rem}
\begin{Hyp}
We expect that the families of functions, listed in theorems \ref{centr-Skl}, \ref{centr-1100} and remarks \ref{centr-1110}, \ref{centr-11-10}, generate the corresponding centers. Taking into account Remark \ref{Rem1.6}, it suffices to check that the rank of the Poisson tensor in these cases is not less than $n^2-n$. For small $n$ this can be verified directly.
\end{Hyp}

\section{Centers of the linear brackets}
Let us recall the general construction of commutative Poisson families by the argument shift method (see, for example\cite{magri1}).
Let $c(L)$ be an element of the center of a quadratic Poisson bracket $\{\cdot, \cdot\}$ of the form \eqref{PoissQuad2}. Then $c(L+\lambda E)$ belongs to the center of the pencil of compatible brackets 
\bea
\label{pencil1}
\{\cdot, \cdot \}_{\lambda}=\{\cdot, \cdot \}+\lambda \{\cdot, \cdot \}_{lin},
\eea 
obtained from \eqref{PoissQuad2} by the argument shift. 
\begin{Lem}
\label{arg_shift}
Suppose that $c(L+\lambda E)$ is a polynomial in $\lambda$:
\bea
\label{Polyn}
c(L+\lambda E)=\sum_{i=0}^m c_i(L) \lambda^{i}.
\eea
Then,
the coefficients $c_i$ of this family commute with each other with respect to both quadratic and linear brackets. In this case, the function $c_m$ belongs to the center of the linear bracket  $\{\cdot, \cdot \}_{lin}$ 
\end{Lem}
\begin{Rem}
\label{sub_dom}
The statement of Lemma \ref{arg_shift} can be clarified: the highest non-constant coefficient of the expansion is central.
It is also known that if there are two central elements of the bracket \eqref{pencil1} polynomial in  $\lambda$, then the entire set of their coefficients generates a commutative Poisson subalgebra with respect to both brackets.
\end{Rem}

In the case where $c(L)$ has the form \eqref{casim}, the above considerations in many cases make it possible to express the central elements of the linear bracket in terms of the coefficients of antidiagonal minors of the matrix $L+\lambda E.$

Consider the case $z_1=z_2=1,z_3=z_4=0$, for which the linear bracket corresponds to the expansion \eqref{bn}. If $k>n/2$ the elements of the center of the quadratic bracket described in Theorem \ref{centr-1100} , after shifting the argument, become polynomials in $\lambda$ of the form
 \bea\label{KazLam}
\frac {\Delta_k^-}{\Delta_{n-k}^-}(L+\lambda E)=\lambda^{2k-n}+\lambda^{2k-n-1}\Delta_k^*+\ldots,\qquad \quad \Delta_k^+ \Delta_{n-k}^+(L+\lambda E)=\lambda^{2k-n}(\Delta_{n-k}^+)^2+\ldots.
\eea
This implies
\begin{Lem}
\label{Centr_lin_1100}
 In the case  $z_1=z_2=1,z_3=z_4=0$ the coefficients $\Delta_k^* $ and antidiagonal minors 
$\Delta_{n-k}^+ $ for  $k>n/2$ belong to the center of the linear bracket.
\end{Lem}
All coefficients of decompositions  \eqref{KazLam} with different $k$ form a commutative Poisson subalgebra for linear and quadratic brackets. Note that for $k=n$, the first decomposition yields the coefficients of the characteristic polynomial of the matrix $L$.

Here is a similar calculation for the Sklyanin bracket, based on   Theorem \ref{centr-Skl}.
We introduce the following notation for the expansion coefficients for $k\le n/2$:
\bea
\frac {\Delta_{n-k}^+} {\Delta_{k}^-}(L+\lambda E)=\lambda^{n-2k}\frac {\Delta_{k}^+}{\Delta_{k}^-}+\lambda^{n-2k-1} \delta_k^{+}+\ldots\nn
\eea
and 
\bea
\frac {\Delta_{n-k}^-}{\Delta_k^+} (L+\lambda E)=\lambda^{n-2k}\frac {\Delta_{k}^-}{\Delta_{k}^+}+\lambda^{n-2k-1} \delta_k^{-}+\ldots.\nn
\eea
\begin{Lem}  
\label{Centr_lin_1000}
For $k\le n/2$, the expressions $\Delta_k^+/\Delta_k^-$ and $\delta_k^-\Delta_k^+/\Delta_k^-+\delta_{k}^+\Delta_k^-/\Delta_k^+$ belong to the center of the linear bracket (\ref{PoissLin2}), (\ref{r1_2}).
\end{Lem}
\begin{proof}
The statement for expressions $\Delta_k^+/\Delta_k^-$ follows from Lemma \ref{arg_shift}. The fact that $\delta_k^-\Delta_k^+/\Delta_k^- +\delta_{k}^+\Delta_k^-/\Delta_k^+$ belongs to the center is a consequence of the Remark \ref{sub_dom} applying to the 
decomposition of
\bea
\frac {\Delta_{n-k}^+} {\Delta_{k}^-}(L+\lambda E) \frac {\Delta_{n-k}^-}{\Delta_k^+} (L+\lambda E)\nn.
\eea
\end{proof}

\medskip
The result of Lemma \ref{Centr_lin_1100} can also be obtained by methods that are used to construct an involutive family of integrals of the full Toda system \cite{DLNT}. Let us first consider a statement about the analogue of the commutative Gelfand-Zeitlin subalgebra for antidiagonal minors. Consider a family of nested Lie subalgebras
\bea
\mathfrak{gl}_n^{(1)}\subset \mathfrak{gl}_n^{(2)}\subset \ldots \subset \mathfrak{gl}_n^{(n-1)}\subset \mathfrak{gl}_n^{(n)}=\mathfrak{gl}_n,\nn
\eea
in which $\mathfrak{gl}_n^{(k)}$ is a subalgebra formed by the upper-right corner of the Lax operator, that is, by generators $e_{ij}$ with 
\bea
 i\le k,\quad j\ge n-k+1.\nn
\eea
Let's also consider the upper-right $k\times k$ corner of the operator $L-\lambda E$ and its determinant
\bea
\Delta_k^+(\lambda)=det((L-\lambda E)^{(k)}).
\eea
\begin{Th}\label{Th5.3}
Coefficients $\Delta_k^+(\lambda)$ lie in the center of $Z(S(\mathfrak{gl}_n^{(k)}))$ and therefore collectively form a commutative family with respect to the standard Kirillov-Kostant bracket.
\end{Th}
\begin{proof}
This statement is most likely well known, but we present its proof for methodological reasons. First of all, we note that
\bea
\label{Lax_action}
\{L-\lambda E, e_{kl}\}_{\mathfrak{gl}}=[L-\lambda E, E_{lk}].
\eea
Here, in the left part of the equation, the Kirillov-Kostant Poisson bracket is calculated for $\mathfrak{gl}_n$ matrix entries of the Lax operator with generators of Poisson algebra, and in the right part stands the commutator of the same Lax operator with matrix unity, as a matrix.  By virtue of (\ref{Lax_action}) the function $\Delta_k(\lambda)$ belongs to the center $Z(S(\mathfrak{gl}_n^{(k)}))$ iff  this function is invariant with respect to the adjoint action of a group of block matrices of the form
\bea
g=\left(\begin{array}{ccc}
1 & 0 & 0\\
\ast & B & 0 \\
\ast & \ast & 1
\end{array}\right); \qquad 
g^{-1}=\left(\begin{array}{ccc}
1 & 0 & 0\\
\ast & B^{-1} & 0 \\
\ast & \ast & 1
\end{array}\right),\nn
\eea
where the upper-left and lower-right blocks are of size $n-k$, and the middle one is of size $2k-n.$
In this case, the upper-right block of the Lax operator of size $k$ is transformed by the rule
\bea
(L-\lambda E)^{(k)}\rightarrow \left(\begin{array}{cc}
1 & 0 \\
\ast & B 
\end{array}\right)(L-\lambda E)^{(k)}\left(\begin{array}{cc}
B^{-1} & 0 \\
\ast & 1 
\end{array}\right).
\eea
This ensures that $\Delta_k^+(\lambda)$ are invariant with respect to such an action.
\end{proof}

\begin{Rem}
\label{inv_uni}
The same calculation is true if we choose $g$ to be a lower triangular unipotent matrix. This proves the invariance of the functions under consideration with respect to the upper triangular nilpotent subalgebra or the corresponding Lie subgroup.
\end{Rem}

Theorem \ref{Th5.3} and Remark \ref{inv_uni} has two curious consequences:
\begin{Lem}
\label{Lem_Nil}
$\Delta_k^+$ for $k\le n/2$ belong to the center of  $S(\mathfrak{n}_+)$.
\end{Lem}  
\begin{proof}
First of all, we note that for such values of the integer $k$, the function $\Delta_k^+(\lambda)$ does not depend on $\lambda$. The statement itself is a consequence of the fact that $\Delta_k^+$ lie in the centralizer $S(\mathfrak{n}_+)$ and, moreover, for $k\le n/2$ they themselves are elements of this algebra.
\end{proof}

Remark \ref{inv_uni} allows us to find the center of the field of fractions $\mathcal{F}S(\mathfrak{b}_+)$
\begin{Lem}
\label{Lem_Bor} 
Let us denote by $\Delta_k^0$ and $\Delta_k^1$ the coefficients in
the expansion of $\Delta_k^+(\lambda)$ for $k>n/2$ on the parameter $\lambda$:
\bea
\Delta_k^+(\lambda)=\lambda^{2k-n}\Delta_k^0+\lambda^{2k-n-1} \Delta_k^1+\ldots\nn
\eea
Then 
\bea
\Delta_k^1/\Delta_k^0\in Z(\mathcal{F}S(\mathfrak{b}_+)).\nn
\eea
\end{Lem}
\begin{proof}
It is easy to see that $\Delta_k^0=\Delta_{n-k}^+$,  both coefficients $\Delta_k^0$ and $\Delta_k^1$ lie in $S(\mathfrak{b}_+),$ and commute with elements of $\mathfrak{n}_+$ by virtue of Remark \ref{inv_uni}. In addition, we find the adjoint action of the Cartan subalgebra on the coefficients of the polynomial
 $\Delta_k^+(\lambda)$. We have:
\bea
\Delta_k ^+(H(L-\lambda E)H^{-1})=\frac {h_1\ldots h_{n-k}}{h_{k+1}\ldots h_n} \Delta_k^+(\lambda).\nn
\eea
The entire polynomial is multiplied by a scalar, hence the ratios of its coefficients are invariant with respect to the action of the Cartan subalgebra.
\end{proof}
\begin{Rem}
The expression $\Delta_k^1/\Delta_k^0$ exactly coincides with the function $\Delta_k^*,$ defined by the formula \eqref{KazLam}. Thus, the statement of Lemmas \ref{Lem_Nil} and \ref{Lem_Bor} is equivalent to the statement of Lemma
\ref{Centr_lin_1100}.
\end{Rem}

\section*{Concluding remarks}
A characteristic property of the family of compatible quadratic Poisson brackets considered in the paper is the 
log-canonicity of brackets between antidiagonal minors and generators of the Poisson algebra. We assume that the log-canonicity conditions for other sets of functions may lead to meaningful examples of families of compatible quadratic Poisson brackets. It is possible that an interesting choice of such sets of functions is related to natural algebraic factorization problems, in particular to some Manin triples.

In the framework of possible future research, we are particularly curious about the quantization of the results obtained, which concerns the quantization of a compatible families of quadratic brackets \eqref{quad-bracket}, \eqref{coe}, their linear incarnations, the centers of both brackets, and commutative families of functions.

\begin{ack} 
The authors are grateful to D. Gurevich, V. Rubtsov, M. Semenov-Tian-Shansky and Yu. Suris for useful discussions.  
Parts 1 and 2 of the work was carried out with the support of the Russian Science Foundation grant 23-71-10110 within the framework of a development programme for the Regional Scientific and Educational Mathematical Center of the Yaroslavl State University with financial support from the Ministry of Science and Higher Education of the Russian Federation (Аgreement on provision of subsidy from the federal budget No. 075-02-2024-1442). 
Parts 3 and 4 of the work were carried out with the support of the Ministry of Education and Science, project No. FSMG-2024-0048.
\end{ack}

\end{document}